\documentclass[12pt]{article}
\usepackage{amsmath,epsfig,amsfonts}
\usepackage{comment}
\usepackage{geometry}
\usepackage{float}

\RequirePackage[OT1]{fontenc}
\RequirePackage{amsthm,amsmath,amssymb}
{
      \theoremstyle{plain}
      \newtheorem{assumption}{Assumption}
  }
\RequirePackage{natbib}
\RequirePackage[colorlinks,citecolor=blue,urlcolor=blue]{hyperref}

\usepackage{moresize}
\usepackage{scalefnt}
\usepackage{multirow}
\usepackage{mathrsfs}
\usepackage{array}
\usepackage[flushleft]{threeparttable}
\usepackage{graphicx}

\newtheorem{thm}{Theorem}

\newcommand\udex[1]{^{\raisebox{1.2pt}{\mbox{$\nano #1$}}}}

\def\var{\mathrm{var}}
\def\cov{\mathrm{cov}}

\def\trans{^{\mbox{\tiny{\sf T}}}}

\def\real{\mathbb R}

\def\indep{\perp \hspace{-0.22cm} \perp}

\def\bop{O_{\nano P}}

\def\nano{\scriptscriptstyle}
\def\inv{^{\nano -1}}
\newcommand\lo[1]{_{\nano #1}}
\def\tint{\textstyle\int}

\def\tsum{\textstyle\sum}

\def\one{I}
\def\two{I\!I}
\def\three{I\!I\!I}
\def\four{I\!V}

\def\logit{\mbox{logit}}
\def\spn{{\mathcal S}}

\newcommand\sam[1]{\udex { (#1)}}

\begin{document}
\baselineskip=22.5pt


\begin{center}
\hspace{-2mm}
{\large \bf  Matching Using Sufficient Dimension Reduction for Causal Inference}
\end{center}

\vspace{2mm}



\begin{center}
Wei Luo$^1$ and Yeying Zhu$^2$\\
$^1$ Department of Statistics and Computer Information Systems, Baruch College\\
$^2$ Department of Statistics and Actuarial Science, University of Waterloo\\ 
\end{center}

\begin{abstract}
To estimate casual treatment effects,  we propose  a new matching approach based on the reduced covariates obtained from sufficient dimension reduction. Compared to the original covariates and the propensity score, which are commonly used for matching in the literature, the reduced covariates are estimable nonparametrically under a mild assumption on the original covariates, and are sufficient and effective in imputing the missing potential outcomes. Under the ignorability assumption, the consistency of the proposed approach requires a weaker common support condition. In addition,  the researchers are allowed to use different reduced covariates to find matched subjects for different treatment groups. We develop relative asymptotic results, and conduct simulation studies as well as real data analysis to illustrate the usefulness of the proposed approach.

\vspace{3mm}

\noindent {\bf Key words and phrases:}
Causal inference; Central subspace; Common support condition; Dimension reduction; Matching
\end{abstract} 

\newpage

\section{Introduction}
One of the commonly used approaches to the causal analysis of observational data is matching, which is a systematic way to find comparable treated and control subjects that have identical or approximate values on an appropriate function of the covariates (\citealp{rosenbaum1983}). The causal effects are then estimated from the matched dataset as if the treatment were randomly assigned. Given a function of the covariates, multiple methods have been developed in the literature to assign matched subjects, e.g.,  nearest available matching based on propensity scores, Mahalanobis metric matching (\citealp{rosenbaum1985}), full matching \citep{rosenbaum1991, hansen2004}, and genetic matching \citep{diamond2013}. The latter two are data-adaptive  in the sense that they both optimize a certain balance criterion based on the data.

Besides the choice of the matching methods, it is also crucial to select an appropriate function of the covariates that these methods are built on, as different choices can dramatically affect the asymptotic behaviors of the resulting estimators \citep{heckman1997, abadie2006}. A natural choice of such functions is the original covariates, which has been commonly used in the literature \citep{rosenbaum1983}. In addition to its ease of use, it ensures the unbiasedness of matching in the population level, under the ignorability assumption and the strong common support condition reviewed in Section $2$ below. However, as pointed out in \citet{abadie2006}, this choice suffers from the ``curse of dimensionality", which makes matching ineffective when the dimension of the covariates is relatively large. 

To address the above-mentioned issue, lower-dimensional functions of covariates have been used for matching in the literature, among which the most popular choice is the propensity score (e.g., \citeauthor{rosenbaum1985}, \citeyear{rosenbaum1985}), defined as the probability of a subject being treated given its covariates.
Same as the original covariates, the propensity score also ensures the unbiasedness of matching under the ignorability assumption and the strong common support condition.
However, the estimation of  the propensity score commonly relies on parametric models, which is susceptible to model mis-specification.
Recently, \citet{de2011} considered the variable selection issue in causal inference, by assuming that a small subset of covariates, called the active set, is sufficient in predicting the outcomes. Using the active set for matching, 
the resulting estimator is asymptotically unbiased, and is effective in finite samples.



In all the existing matching procedures, the same function of covariates is used to find matches in different treatment groups. When the outcome regression pattern differs, which in \citet{de2011}, means varying active sets in different treatment groups, these sets will be merged into one  set for subsequent matching. The merged active set will then always contain noisy information for both treatment groups. In this case, it is conceivable that the effectiveness of matching can be enhanced, if we allow for different functions of covariates to be used for different treatment groups.

In this article, we introduce a new set of functions of the covariates with the aid of sufficient dimension reduction (SDR, \citeauthor{li1989},\citeyear{li1989}; \citeauthor{li1991},\citeyear{li1991}). Unlike the existing approaches, the functions of covariates vary  in different treatment groups to find matches. Under a mild assumption, the new functions enjoy the advantages of both the original covariates and the propensity scores. That is, they are easily estimable in a model-free manner, and are usually of low dimension in practice. In particular, their dimensions are always equal to or lower than the cardinality of the active set in \citet{de2011}.
In addition, as seen later, the new functions help relax the strong common support condition, and thus enhance the applicability of the matching methods.

The structure of this article is as follows. In Section $2$, we will briefly review the potential outcomes framework and the literature of SDR, with an emphasis on sliced inverse regression (SIR, \citeauthor{li1991}, \citeyear{li1991}). We then propose a new matching estimator using SDR in Section $3$, and develop some relative asymptotic results in Section $4$. Section $5$ includes the details in implementation. Simulation studies and real data analysis are conducted in Section $6$ and $7$ to illustrate the proposed estimator. We will leave more discussion to Section $8$.

\section{A review of literature}

\subsection{Potential Outcomes Framework}

Matching can be explained under the commonly used potential outcomes framework (\citealp{neyman1923, rubin1974}) that assumes the complete data to be $n$ copies of $(X, T, Y(0), Y(1))$, denoted as $\{(X\sam i, T\sam i, Y\sam i (0)$, $Y\sam i (1)), i=1,\ldots n\}$. Here,  $X \in \mathbb{R}\udex p$ is the set of covariates with support $\Omega (X)$; $T$ is assumed to be the binary treatment assignment with $T=0$ being the control group and $T=1$ being the treatment group; and $Y(t)  \in \real$ is the potential outcome under treatment $t$ and is observed only when $T=t$. The potential outcome framework assumes the existence of the hypothetical outcome under the treatment level that is not observed in the data. Different types of causal effects are defined depending on the research interest. For example, the average causal effect is defined as
\begin{align*}
ACE=E[Y(1)-Y(0)],
\end{align*}
and the average causal effect among the treated is defined as
\begin{align*}
ACET=E[Y(1)-Y(0)|T=1].
\end{align*}

To estimate these parameters from the observed data, researchers commonly regulate the missingness mechanism by the ignorability assumption (\citealp{rosenbaum1983}),
\begin{align}\label{assumption-ig}
Y(t) \indep T \,|\, X, \quad t = 0, 1
\end{align}
in which $\indep$ means independence. Under (\ref{assumption-ig}), \citet{rosenbaum1983} introduced a family of measurable functions of $X$ called the balancing score, denoted by $R(X)$,  which satisfies
\begin{align}\label{assumption-bal-score}
Y(t) \indep T \,|\, R (X),
\end{align}
for $t=0, 1$. Examples of $R(X)$ include the covariates $X$ and the aforementioned propensity score.

A natural interpretation of (\ref{assumption-bal-score}) is that if we find two subjects that share the same value of $R(X)$, then no matter whether they are from the same treatment group or not, their outcomes $Y(t)$ should  have identical distributions. Consequently, if we manage to match each subject with those in the alternative treatment group according to their similarity in $R(X)$, and impute its missing potential outcome using the observed outcomes in the matched set, then we can treat the imputed dataset as observed completely, and estimate the aforementioned causal effects using the usual sample moments. This strategy is used in the Abadie-Imbens estimator (\citealp{abadie2006}), which imputes the missing potential outcome of subject $i$ by
\begin{equation}\label{eq:impute}
\hat{Y} \sam i(1-T\sam i)=\frac{1}{m} \sum_{j\in J\lo m (i)} Y\sam j (T\sam j),
\end{equation}
in which $J\lo m(i)$ is the matched set for subject $i$ with cardinality $m$.
Clearly, $J\lo m (i)$ is determined by $R(X)$, and is crucial to the consistency of matching. For a fixed value of $m$, when a high-dimensional $R(X)$ is used, such as $X$, the subjects in $J\lo m (i)$ may differ dramatically in terms of $R(X)$. This may lead to a large bias in  (\ref{eq:impute}). When $R(X)$ fails to satisfy (\ref{assumption-bal-score}), which occurs if the propensity score is inconsistently estimated, then even if the subjects in $J\lo m (i)$ share the same value of $R(X)$ as subject $i$,
their outcomes can be stochastically different
and lead to a large bias again. These are the drawbacks of using the original covariates and  the estimated propensity score for matching.

To ensure the consistency of matching using $R(X)$, the data must satisfy that
$P(T = 1 | R(X)) \in (0, 1)$ almost surely, which, by Bayes' Theorem, amounts to the coincidence of the support of $R(X)$ in the treatment and control groups.
When $R(X)$ is either the covariates or the propensity score, the condition is further equivalent to the common support of the covariates across the treatment groups, which we refer to as the strong common support condition.  The condition can  be relaxed if other forms of $R(X)$ are used,  for example, if we employ SDR to construct a low-dimensional balancing score $R(X)$, other than the propensity score, for matching. We review the basics of SDR in the following subsection. 

\subsection{Sufficient dimension reduction}

SDR has been widely used to reduce the dimension of the covariates prior to data analysis. For the covariates $X$ and a response variable $W$, it assumes the existence of $\beta \in \real\udex {p\times d}$, where $d < p$, such that the reduced covariates $X'\beta$ are sufficient for the subsequent modeling of $W$. That is,
\begin{equation}\label{assumption-sdr}
W \indep X \, | \, X'\beta.
\end{equation}
For identifiable parametrization, \citet{cook1998} introduced the central subspace, denoted by $\spn \lo {W | X}$, as the column space spanned by $\beta$ in (\ref{assumption-sdr}) with the smallest dimension $d$. Conditions required for existence of this space are fairly general, and are adopted throughout the article. If $\beta$ in (\ref{assumption-sdr}) is further restricted so that each of its column must have exactly one nonzero component, then $X'\beta$ becomes the active set in variable selection. Thus, by definition, the dimension of the central subspace must be equal to or less than the cardinality of the active set. In this sense, SDR can be more effective than variable selection in reducing the dimension of the covariates, subject to consistent estimation of the central subspace.

In the literature, multiple methods have been developed to estimate the central subspace, including sliced inverse regression (SIR, \citeauthor{li1991},\citeyear{li1991}), sliced average variance estimation (SAVE, \citeauthor{cook1991},\citeyear{cook1991}) and the directional regression (\citealp{li2007}), among which SIR is commonly used for its simplicity and finite sample effectiveness. For consistency, it requires the {\it linearity condition} on $X$:
\begin{equation}\label{assumption-linearity}
E[\{X - E(X)\} | X'\beta] = \Sigma\lo X \beta (\beta '\Sigma\lo X \beta)\inv \beta' \{X - E(X)\},
\end{equation}
where $\Sigma\lo X$ is the covariance matrix of $X$. The condition includes the elliptical distribution of $X$ as a special case, and holds approximately for general high-dimensional covariates (\citealp{hall1993}). Thus, the linearity condition is not considered restrictive in applications. As no other model structure is assumed, SIR is termed model-free. More details can be found in Section $5$ and \citet{li1991}.

A separate issue in SDR is to determine the dimension of the central subspace. In the literature, the sequential testing procedure (\citealp{li1991}) has been shown effective for this purpose. For continuity of the context, we leave its detail to the Appendix A (Supplementary Material).

\section{Matching based on SDR}

From the discussions above, when the SDR assumption (\ref{assumption-sdr}) holds in the data with the treated and control outcomes to be the response variable $W$, respectively, the reduced covariates $X'\beta$ are low-dimensional, and are estimable nonparametrically had the outcomes been completely observed. Therefore, $X'\beta$ will be desired functions of covariates to be matched on if they are balancing scores in (\ref{assumption-bal-score}). This motivation is justified in the following theorem.

\begin{thm}\label{thm:sdr-bal}
For $t=0, 1$, suppose that the central subspace $\spn \lo {Y(t) | X}$ is $r(t)$-dimensional, where $r(t) < p$. Then under the ignorability assumption, its arbitrary basis matrix $\beta\lo t \in \real\udex {p \times r(t)}$ generates a balancing score $X'\beta\lo t$ that satisfies (\ref{assumption-bal-score}) for $t=0,1$.
\end{thm}

\begin{proof}
Without loss of generality, we denote $f(y(t) | X, T=t)$ as the density function of the conditional distribution $Y(t) | (X, T =t)$ with respect to a $\sigma$-finite measure $\nu$, for $t=0, 1$. For any random element $R$, denote $\Omega(R | T=t)$ and $\Omega(R)$ as the support of $R | T=t$ and $R$, respectively. For any $a \in \Omega (X'\beta\lo t)$, if $a \not\in \Omega (X'\beta\lo t | T=t)$, then $X'\beta\lo t = a$ implies that $T = 1-t$ almost surely, which automatically implies that
\begin{equation}\label{eq:prf-thm-bal-1}
Y(t) \indep T \,|\, X'\beta\lo t = a.
\end{equation}
Thus we only need to show (\ref{eq:prf-thm-bal-1}) for any $a \in \Omega (X'\beta\lo t | T=t)$. That is, $f(y(t) | X'\beta\lo t = a, T=t) = f ( y(t) | X'\beta\lo t = a)$. By the definition of $\beta\lo t$, we have
\begin{equation*}
Y(t) \indep X \,|\, X'\beta\lo t,
\end{equation*}
which means that $f (y(t) | X) = f (y(t) | X'\beta\lo t)$. For any $x \in \Omega (X | T=t)$ such that $X'\beta\lo t = a$, $f ( y (t) | X = x, T=t) = f (y(t) | X=x) = f ( y(t) | X'\beta\lo t = a)$, where the first equality  is due to the ignorability assumption. Thus we have
\begin{eqnarray*}
f(y(t) | X'\beta\lo t = a, T=t) &=& E \{f(y (t) | X=x, T=t) | X'\beta\lo t =a, T=t \} \\
&=& E \{f( y(t) | X'\beta\lo t = a) | X'\beta\lo t =a, T=t \} \\
&=& f ( y(t) | X'\beta\lo t = a).
\end{eqnarray*}
This completes the proof.
\end{proof}

Based on this theorem, if we can estimate the central subspaces $\spn\lo {Y (0) | X}$ and $\spn\lo {Y(1) | X}$ consistently, say by $\spn (\hat \beta\lo 0)$ and $\spn (\hat \beta\lo 1)$, where $\spn (\cdot)$ denotes the linear span of a matrix, then we can use the reduced covariates $X'\hat \beta\lo 0$ in place of $X$ to find matches for each treated subject, and use $X'\hat \beta\lo 1$ to find matches for each control subject.  

Because we allow the two central subspaces to differ, we may end up using different criteria for different treatment groups in the Abadie-Imbens estimator. As mentioned before, this flexibility is meaningful, because it allows us to use the exactly useful information in the covariates without introducing noise. 

Due to the missingness in the observed data, the central subspaces in Theorem \ref{thm:sdr-bal} are not readily estimable by any aforementioned SDR method. The estimable central subspaces are those conditional on the treatment assignments, denoted by $\spn \lo {Y(0) | X}\udex D$ and $\spn\lo {Y(1) | X} \udex D$, where the superscript ``$D$" refers to data. They are spanned by the basis matrices $\beta\lo 0\udex D$ and $\beta\lo 1\udex D$ that satisfy
\begin{equation*}\label{assumption-sdr-data}
Y(0) \indep X \,|\, (X'\beta\lo 0\udex D, T= 0), \quad Y(1) \indep X \,|\, (X'\beta\lo 1\udex D, T= 1),
\end{equation*}
respectively. Fortunately, under the ignorability assumption and a mild additional condition, these estimable central subspaces are identical to their counterparts of interest, and thus can serve as the substitutes for the desired central subspaces.
In the following, we denote $\Omega(R | T=t)$ and $\Omega(R)$ as the support of $R | T=t$ and the support of $R$, respectively, for any random element $R$ and $t=0, 1$.

\def\csz{\spn\lo {Y(0) | X}}
\def\cso{\spn\lo {Y(1) | X}}
\def\csdz{\spn\lo {Y(0) | X}\udex D}
\def\csdo{\spn\lo {Y(1) | X}\udex D}

\begin{thm}\label{thm:equiv-cs}
Under the ignorability assumption (\ref{assumption-ig}), for $t=0, 1$, suppose that $\Omega(X'\beta\lo t | T=t) = \Omega (X' \beta\lo t)$, then $\spn\lo {Y(t) | X}\udex D = \spn\lo {Y(t)|X}$.
\end{thm}

\begin{proof}
For simplicity, we only prove the theorem for $t=0$. The case for $t=1$ can be shown in the same manner. First, we show that $\csdz \subseteq \csz$, which is equivalent to:
\begin{equation}\label{eq:prf-thm-equi-1}
Y(0) \indep X \,|\, (X' \beta\lo 0, T=0).
\end{equation}
Following the notations in the proof of Theorem \ref{thm:sdr-bal}, we denote $f(\cdot | X)$ as the conditional density function of $Y(0) | X$ and $f(\cdot |X, T=0)$ as the conditional density function of $Y(0) | X$ when $T=0$. (\ref{eq:prf-thm-equi-1}) is equivalent to that $f(\cdot | X=x, T=0) = f(\cdot | X' \beta\lo 0 = x'\beta\lo 0, T=0)$ for any $x \in \Omega(X | T=0)$.  
By the ignorability assumption, $f(\cdot | X=x, T= 0) = f(\cdot | X = x)$, and by the definition of $\csz$, $f(\cdot | X=x) = f(\cdot | X'\beta\lo 0  = x'\beta\lo 0)$. Thus $f(\cdot | X=x, T=0)$ is measurable with respect to $x'\beta\lo 0$, which means that
\begin{eqnarray*}
f(\cdot | X=x, T=0) &=& E\{f(\cdot | X, T=0) | X'\beta\lo 0 = x'\beta\lo 0, T=0\} \\
&=& f(\cdot | X'\beta\lo 0 = x'\beta\lo 0, T=0).
\end{eqnarray*}
Hence (\ref{eq:prf-thm-equi-1}) holds. 
Conversely,
to show that $\csz \subseteq \csdz$, note that it is equivalent to
\begin{equation}\label{eq:prf-thm-equi-2}
Y(0) \indep X \,|\, X' \beta\lo 0\udex D,
\end{equation}
which holds if $f(\cdot | X=x) = f(\cdot | X' \beta\lo 0\udex D = x'\beta\lo 0\udex D)$ for any $x \in \Omega(X)$. Because $\Omega (X'\beta\lo 0 | T=0) = \Omega (X' \beta\lo 0)$, there exists $x\udex * \in \Omega (X | T=0)$ such that $x' \beta\lo 0 = x\udex {*\prime} \beta\lo 0$. By the definition of $\csz$, we have
\begin{equation*}
f(\cdot | X=x) = f(\cdot | X'\beta\lo 0 = x'\beta\lo 0) = f(\cdot | X' \beta\lo 0 = x\udex {*\prime} \beta\lo 0) =f (\cdot | X = x\udex *).
\end{equation*}
Since $x\udex * \in \Omega (X | T=0)$, by the ignorability assumption, $f(\cdot | X\udex *=x\udex *) = f(\cdot | X\udex * = x\udex *, T=0)$, which, by the definition of $\csdz$, further implies that $f(\cdot | X\udex *=x\udex *) = f(\cdot | X\udex {*\prime} \beta\lo 0\udex D = x\udex {*\prime} \beta\lo 0\udex D, T=0)$. Thus $f(\cdot | X\udex *=x\udex *)$ is measurable with respect to $x\udex {*\prime} \beta\lo 0\udex D$, which, similar to the above, implies that $f(\cdot | X\udex *=x\udex *) = f(\cdot | X\udex {*\prime} \beta\lo 0\udex D = x\udex {*\prime} \beta\lo 0\udex D)$.  Because $\csdz \subseteq \csz$, we have $x' \beta\lo 0\udex D = x\udex {*\prime} \beta\lo 0\udex D$. Hence $f(\cdot | X=x) = f(\cdot | X' \beta\lo 0\udex D = x'\beta\lo 0\udex D)$, which implies (\ref{eq:prf-thm-equi-2}). This completes the proof.
\end{proof}

The condition in this theorem requires that no matter whether the subject is known to be treated, the support of the covariates is invariant in the directions that are informative to the treated outcome, and likewise for the control outcome. As the condition does not impose any restriction in the directions of covariates that are redundant to the outcomes, it relaxes the strong common support condition mentioned in Section $2$, and is referred to as the weak common support condition. A similar result for causal inference based on the outcome regressions can be found in \citet{luo2016}.

The theorem justifies that we can estimate the interested central subspaces using the observed data. Various SDR methods can be employed. We choose SIR as an example, and describe its implementation in detail in Section $5$. Given the resulting estimates, denoted by $\hat \beta\lo 0$ and $\hat \beta\lo 1$, respectively, we use $X'\hat \beta\lo 0$ and $X' \hat \beta\lo 1$ in the subsequent matching methods mentioned in the Introduction. As discussed above, the proposed matching approach outperforms the existing approaches based on either the original covariates or the propensity score in multiple ways. These advantages are further formulated in an asymptotic sense in Section $5$.

For the consistency of SIR, the covariates in each treatment group must satisfy the linearity condition (\ref{assumption-linearity}), with $\Sigma\lo X$ being replaced with the conditional covariance matrix $\Sigma\lo {X|T}$. We allow $\Sigma\lo {X|T}$ to be stochastic, so that the proposed estimator is applicable, for example, when the covariates have different correlation structures in the treatment and control groups. When the reduced covariates in each group are elliptically distributed with a non-stochastic $\Sigma\lo {X|T}$ and equipped with an affinely invariant measure in the subsequent matching, it may be of interest to see whether the proposed approach is conditionally affinely invariant on the original covariates, which would then imply the conditionally equal percent bias reduction property (\citealp{rubin1992, rubin1996}).
The answer depends on in which space such invariance is referred to: the approach is conditionally affinely invariant within the estimated central subspace but not within the true central subspace, as the estimated reduced covariates are invariant if the original covariates are rotated within the former, whereas the invariance does not hold if the rotations are in the latter.

\section{Asymptotic study}

Using the Abadie-Imbens estimator (\ref{eq:impute}) on the reduced covariates, we now present some asymptotic results that illustrate the superiority of the proposed approach to those using the propensity score or the original covariates. For ease of presentation, we focus on the estimation of $ACE$. Similar results can be derived for $ACET$ with slight adjustments.

For simplicity, we fix the value of $m$ as the sample size grows for all the balancing scores.
In the literature, its has been commonly realized that the asymptotic property of matching is difficult to tackle, if the estimated propensity score, rather than its true value, is used (\citealp{imbens2015}). Thus, in such studies, the true propensity score is assumed known a priori (\citealp{abadie2006}),  unless a high-order kernel density function is used in kernel matching (\citealp{heckman1997}). Following this convention,
we assume both the propensity score and the central subspaces to be known a priori throughout the section. Although by doing this, we omit the effect of SDR estimation, the results in this section still help illustrate the advantage of using SDR for matching, and provide some guidelines when other SDR methods, such as the aforementioned SAVE and directional regression, are employed. In addition, as seen in the simulation studies, the effect of the SDR estimation is almost negligible in finite samples. For these reasons, we think that the results developed here are meaningful.

For regulation purpose, following \citet{abadie2006}, we adopt several assumptions as below:

\begin{assumption}\label{assumption-regularity}
The sample subjects are independent copies of $(X, T, Y(T))$. The distribution of $X$ is dominated by the Lebesgue measure on $\real\udex p$ with a compact support $\Omega (X)$. $E\{Y(t) | X\}$ is twice differentiable, and its Hessian matrix is Lipschitz continuous almost surely for $t=0, 1$.
\end{assumption}

\begin{assumption}\label{assumption-cmmn-sppt-stronger}
There exists $\tau > 0$ such that the propensity score $\pi(X)$ satisfies that $\pi(X) \in (\tau, 1 -\tau)$ almost surely.
\end{assumption}

Assumption \ref{assumption-regularity} is fairly general, and has been commonly adopted in the literature. Assumption \ref{assumption-cmmn-sppt-stronger} further regulates the strong common support condition in Section $2$. Similarly, we regulate the weak common support condition in the following, which relaxes Assumption \ref{assumption-cmmn-sppt-stronger}.

\begin{assumption}\label{assumption-cmmn-sppt-weak} Let $\pi(X'\beta\lo t) = P(T = 1 | X'\beta\lo t)$ for $t=0, 1$.
There exists $\tau\lo r > 0$ such that $\pi(X'\beta\lo 0) < 1 - \tau\lo r$ and $\pi(X'\beta\lo 1) > \tau\lo r $ almost surely.
\end{assumption}

Under these assumptions, the following theorem shows the effect of the dimensionality of the balancing score on the convergence order of the resulting matching estimator. By conducting a more careful study on the variance of the estimator, it strengthens the result in Theorem 1 of \citet{abadie2006} to allow a higher-dimensional balancing score for $\sqrt n$-consistent estimation.

\def\dy{\Delta Y}

\begin{thm}\label{thm:asym-bias}
Suppose that the ignorability assumption (\ref{assumption-ig}) and Assumption \ref{assumption-regularity} hold. Let $\mu$ be $ACE$ and $\hat \mu\lo {X}$, $\hat \mu\lo {\pi}$ and $\hat \mu\lo {r}$ be the Abadie-Imbens estimator using the original covariates, the true propensity score, and the reduced covariates from the true central subspaces, respectively. Then,

\noindent
(1) under Assumption \ref{assumption-cmmn-sppt-stronger},
\begin{eqnarray*}
&& E(\hat \mu\lo {X}) - \mu = O (n\udex {-2/p}), \mbox{ and }
\var(\hat \mu\lo {X}) = \bop (n\udex {-\min\{1,6/p\}}), \\
&& E(\hat \mu\lo {\pi}) - \mu = O (n\udex {-2}), \mbox{ and }
\var(\hat \mu\lo {\pi}) = \bop (n\udex {-1}).
\end{eqnarray*}

\noindent
(2) under Assumption \ref{assumption-cmmn-sppt-weak} and the linearity condition (\ref{assumption-linearity}) for $X|T$,
\begin{eqnarray*}
E(\hat \mu\lo {r}) - \mu = O (n\udex {-2/\max\{r(0), r(1)\}}), \mbox { and } \var(\hat \mu\lo {r}) = \bop (n\udex {-\min\{1,6/r(0), 6/r(1)\}}).
\end{eqnarray*}
\end{thm}

\begin{proof}
We first show a mathematical property that will be used in the proof. Let $r, s> 0$ be arbitrary positive real numbers. We show that, as $n \rightarrow \infty$,
\begin{eqnarray}\label{eq:lemma-conv-exp}
n \left[\{1 - (r + s) n\udex {-1}\} \udex n - (1 - rn\udex {-1}) \udex n (1 - s n\udex {-1})\udex n\right] \rightarrow - r s \exp\{-(r+s)\},
\end{eqnarray}
which is a complement of the well known result $\lim\lo {n\rightarrow \infty} (1 - r n\udex {-1})\udex n = e\udex {-r}$. To see why it holds, re-write the left-hand side as
\begin{equation*}
n[\{1 - (r + s) n\udex {-1}\} \udex n - \{1 - (r + s) n\udex {-1} + r s n\udex {-2}\}\udex n ],
\end{equation*}
and divide it by $\{1 - (r + s) n\udex {-1}\}\udex n$, which converges to $\exp\{-(r+s)\}$. Then (\ref{eq:lemma-conv-exp}) is equivalent to
\begin{equation*}
n \left[1 -  [1 + rs n\udex {-2} \{1 - (r + s) n\udex {-1}\}\inv]\udex n \right] \rightarrow - r s.
\end{equation*}
Denote the left-hand side above by $\epsilon\lo n$. By simple algebra, 
\begin{equation*}
\lim\lo {n\rightarrow \infty }( 1 - \epsilon\lo n / n)\udex n = e\udex {rs},
\end{equation*}
which means that $\epsilon\lo n \rightarrow - (rs)$, and (\ref{eq:lemma-conv-exp}) holds.
For simplicity, we prove the theorem for $\hat \mu\lo X$ and $m=1$, and denote $J\lo 1(i)$ by $J(i)$ for each subject $i$. The general case can be shown similarly. Denote $n\lo 0$ and $n\lo 1$ as the number of control and treated subjects, respectively. As we regard $T$ to be random, so are $n\lo 0$ and $n\lo 1$.
By the law of large numbers, for $t=0,1$, $n\lo t n\inv \rightarrow (1-t) + (2t-1) P(T=1)$ in probability. Thus, for any $w \in \real$, $\bop (n\lo t\udex w) = \bop (n\udex w)$.
For efficiency of presentation, without loss of generality, we treat each $T\sam i$ as given, as well as $n\lo 0$ and $n\lo 1$, in the rest of the proof.
Following equation (7) in \cite{abadie2006}, let $K(i)$ be the number of times subject $i$ is used to match the others. $\hat \mu\lo X$ can be decomposed as
\begin{eqnarray*}
\hat \mu\lo X - \mu &=& n\inv  \tsum\lo {i=1}\udex n [E\{Y(1) | X\sam i\} - E\{Y(0) | X\sam i\} - \mu] \\
&+& \hspace{0.2cm} n\inv  \tsum\lo {i=1}\udex n (2 T\sam i - 1) \{1 + K(i)\}[Y\sam i (T\sam i) - E\{Y(T) | X\sam i \}]  \\
&+& \hspace{0.2cm} n\inv  \tsum\lo {i=1}\udex n (2 T\sam i - 1)[E\{Y(1- T) | X\sam i\} - E\{Y(1-T) | X\sam {J (i)}\}],
\end{eqnarray*}
in which the first two terms on the right hand side are easily seen to be $\bop (n\udex {-1/2})$ with mean zero (\citealp{abadie2006}). The third term represents the bias in matching, which we denote by $B\lo n$. To show the result about $\hat \mu\lo X$, it suffices to show that $E(B\lo n) = O(n\udex {-2/p})$ and $\var (B\lo n) = O(n\udex {-\min\{1+2/p,6/p\}})$. By Assumption \ref{assumption-regularity} and Taylor's expansion,
\begin{eqnarray*}
&& E\{Y(1- T) | X\sam i\} - E\{Y(1-T) | X\sam {J (i)}\} = (X\sam i - X\sam {J(i)})' g\lo i\sam i \\
&& \hspace{1.5cm} + \, (X\sam i - X\sam {J(i)})' g\lo 2\sam i (X\sam i - X\sam {J(i)}) + g\lo 3\udex {(i,J(i))} \|X\sam i - X\sam {J(i)}\|\udex 3,
\end{eqnarray*}
where $g\lo 1\sam i$ and $g\lo 2 \sam i$ are the gradient and the Hessian matrix of $E\{Y(1-T) \mid X\}$ at $X\sam i$, and $g\lo 3\udex {(i, J(i))}$ is a bounded random element.
Thus, for the desired result about $E(B\lo n)$ and $\var (B\lo n)$, it suffices to show that, for an arbitrary pair of subjects $(i,j)$,
\begin{eqnarray}\label{eq:prf-thm-bias-1}
&& X\sam {J(i)} - X\sam i = \bop (n\udex {-1/p}), E(X\sam {J(i)} - X\sam i) = O(n\udex {-2/p}), \nonumber \\
&& \cov \{(X\sam {J(i)} - X\sam i)\udex {\otimes k}, (X\sam {J(j)} - X\sam j)\udex {\otimes l}\} = O(n\udex {-1-2/p}).
\end{eqnarray}
in which $k, l \in  \{1, 2\}$, and $v\udex {\otimes 1} = v$ and $vv\udex {\otimes 2} = vv\trans$ for any real vector $v$.
To show (\ref{eq:prf-thm-bias-1}), we first suppose that both $i$ and $j$ belong to the same treatment group $T=t$.
For any $a, b \in \Omega(X |T = t)$, let $(u,v) = n\lo {1-t}\udex {1/p} (X\sam {J(a)} - a, X\sam {J(b)} - b)$, in which $J(a)$ denotes the subject whose covariates are closest to a, and $J(b)$ likewise.
From the proof of Theorem 1 in \citet{abadie2006}, let $f(\cdot)$ be the density function of random elements measurable with respect to $\{X\sam i, i=1,\ldots, n\}$, we have

{\small {
\begin{eqnarray}\label{eq:prf-thm-bias-2}
f(u) &=& f(a + u n\lo {1-t}\udex {-1/p}) \{1 - P (\|X - a \| \leq \|u\| n\lo {1-t}\udex {-1/p} )\}\udex {n\lo {1-t}-1} \nonumber \\
&=& \{f(a) + f\udex {*\prime} (a, u) u n\lo {1-t}\udex {-1/p}\} \{1 - P (\|X - a \| \leq \|u\| n\lo {1-t}\udex {-1/p} )\}\udex {n\lo {1-t} - 1},
\end{eqnarray}}}
in which $f\udex *(a, u)$ is defined so as to make the equation holds. By Assumption $1$, $f\udex *(a, u)$ is bounded. As shown in \cite{abadie2006}, $\{1-P (\|X - a \| \leq \|u\| n\lo {1-t}\udex {-1/p} )\}\udex {n\lo {1-t}}$ converges to $\exp [- 2 \pi\udex{p/2}\|u\|\udex p f(a) / \{p\Gamma(p/2)\}]$. Thus $u = \bop (1)$, which implies that $X\sam {J(a)} - a = \bop (n\udex {-1/p})$. Using the symmetry of $\{1-P (\|X - a \| \leq \|u\| n\lo {1-t}\udex {-1/p} )\}\udex {n\lo {1-t}}$ about the origin in $\real\udex p$, \citet{abadie2006} further showed that
$E(X\sam {J(a)} - a) = O(n\udex {-2/p})$. By the compactness of $\Omega(X)$, 
such convergence is uniform for $a$, thus $X\sam {J(i)} - X\sam i = \bop (n\udex {-1/p})$ and $E(X\sam {J(i)} - X\sam i) = O(n\udex {-2/p})$. Similar to (\ref{eq:prf-thm-bias-2}), 
we further have
\begin{eqnarray}\label{eq:prf-thm-bias-3}
f(u,v) &=&
\{n\lo {1-t}(n\lo {1-t}-1)\} n\lo {1-t} \udex {-2} f(a + u n\lo {1-t}\udex {-1/p}) f(b + v n\lo {1-t}\udex {-1/p}) \nonumber \\
&& \{1 - P (\|X - a \| \leq \|u\| n\lo {1-t}\udex {-1/p} \mbox{ or } \|X - b\| \leq \|v\| n\lo {1-t}\udex {-1/p})\}\udex {n\lo {1-t}-2} \nonumber
\end{eqnarray}
in which $P (\|X - a \| \leq \|u\| n\lo {1-t}\udex {-1/p} \mbox{ or } \|X - b\| \leq \|v\| n\lo {1-t} \udex {-1/p})$ can be written as
\begin{eqnarray*}
&& P (\|X - a \| \leq \|u\| n\lo {1-t} \udex {-1/p}) + P(\|X - b\| \leq \|v\| n\lo {1-t} \udex {-1/p}) \\
&& \hspace{0.3cm} - \, P (\|X - a \| \leq \|u\| n\lo {1-t}\udex {-1/p} \mbox { and } \|X - b\| \leq \|v\| n\lo {1-t}\udex {-1/p}) \\
&& \equiv \one + \two + \three
\end{eqnarray*}
Obviously, $\three = 0$ for all large $n$. Let
$r\lo n = n\lo {1-t} P (\|X - a \| \leq \|u\| n\lo {1-t} \udex {-1/p})$, 
and $s\lo n = n\lo {1-t} P (\|X - b \| \leq \|v\| n\lo {1-t} \udex {-1/p})$.
Then
\begin{eqnarray*}
f(u,v) = (1-n\lo {1-t} \udex{-1}) f(a + u n\lo {1-t}\udex {-1/p}) f(b + v n\lo {1-t}\udex {-1/p}) (1-r\lo n n\lo {1-t}\inv - s\lo n n\lo {1-t}\inv)\udex {n\lo {1-t}-2}
\end{eqnarray*}
for all large $n$, and $r\lo n\rightarrow r$ and $s\lo n\rightarrow s$, in which $r = 2 \pi\udex{p/2}\|u\|\udex p f(a) / \{p\Gamma(p/2)\}$ and $s = 2 \pi\udex{p/2}\|v\|\udex p f(b) / \{p\Gamma(p/2)\}$.
Since
\begin{eqnarray*}
f(u)f(v) &=& f(a + u n\lo {1-t}\udex {-1/p}) f(b + v n\lo {1-t}\udex {-1/p})
(1-r\lo n n\lo {1-t}\inv)\udex {n\lo {1-t} - 1}(1-s\lo n n\lo {1-t}\inv)\udex {n\lo {1-t} - 1},
\end{eqnarray*}
we have, for $k,l \in \{1, 2\}$,
\begin{eqnarray*}
&& E\{u\udex {\otimes {k}} (v\udex {\otimes l})'\} - E(u\udex {\otimes k})E'(v\udex {\otimes l}) \\
&& \hspace{0.5cm} = \tint\lo {\real\udex p\times \real\udex p} u\udex {\otimes {k}} (v\udex {\otimes l})' \{f(u, v) - f(u) f(v)\} du dv \\
&& \hspace{0.5cm} = \tint\lo {\real\udex p\times \real\udex p} u\udex {\otimes {k}} (v\udex {\otimes l})' f(a + u n\lo {1-t}\udex {-1/p}) f(b + v n\lo {1-t}\udex {-1/p}) \eta\lo {a,b}  du dv + O(n\inv),
\end{eqnarray*}
in which $\eta\lo {a,b} = \{1 - r\lo n n\lo {1-t}\inv - s\lo n n\lo {1-t}\inv \}\udex {n\lo {1-t}} - (1 - r\lo n n\lo {1-t}\inv)\udex {n\lo {1-t}} (1-s\lo n n\lo {1-t}\inv)\udex {n\lo {1-t}}$. By (\ref{eq:lemma-conv-exp}), $\eta\lo {a,b} \rightarrow - {n\lo {1-t}}\inv r s \exp \{-(r+s)\}$. Since $\|u\|\udex {2} \|v\|\udex 2 rs \exp \{-(r+s)\}$ is integrable on $(u,v) \in \real\udex p \times \real\udex p$, similar to \cite{abadie2006},
we have $E\{u\udex {\otimes {k}} (v\udex {\otimes l})'\} - E(u\udex {\otimes k})E'(v\udex {\otimes l}) = O(n\udex{-1})$. Next, suppose that $i$ and $j$ are control and treated subjects, respectively. Conditioning on $(X\sam i, X\sam j) = (a,b)$, let $u = n\lo 1\udex {1/p}(X\sam {J(a)} - a)$ and $v = n\lo 0\udex {1/p} (X\sam {J(b)} - b)$, we have
\begin{eqnarray}\label{eq:prf-thm-bias-3}
f(u,v) &=&
I(\max\{n\lo 1\udex {-1/p}\|u\|, n\lo 0\udex {-1/p}\|v\|\} < \|a-b\|) (1-n\lo 1\inv) (1-n\lo 0\inv) \nonumber \\
&& f(a + u n\lo 1\udex {-1/p}) f(b + v n\lo 0\udex {-1/p}) (1 - r\lo n)\udex {n\lo 1 - 2} (1 - s\lo n)\udex {n\lo 0 - 2} \nonumber \\
&& \hspace{-0.2cm} + \, I(n\lo 1\udex {-1/p}\|u\| < \|a-b\|)\delta\lo {\|a-b\|}(\|v\|)(1-n\lo 1\inv) n\lo 0\inv f(a + u n\lo 1\udex {-1/p}) \nonumber \\
&& \hspace{0.2cm} (1 - r\lo n)\udex {n\lo 1 - 2} (1-s\lo n) \udex {n\lo 0 - 1} \nonumber \\
&& \hspace{-0.2cm} + \, I(n\lo 0\udex {-1/p}\|v\| < \|a-b\|)\delta\lo {\|a-b\|}(\|u\|)(1-n\lo 0\inv) n\lo 1\inv f(b + v n\lo 0\udex {-1/p}) \nonumber \\
&& \hspace{0.2cm} (1 - s\lo n)\udex {n\lo 0 - 2} (1-r\lo n) \udex {n\lo 1 - 1} \nonumber \\
&& \hspace{-0.2cm} + \, \delta\lo {(\|a-b\|, \|a-b\|)}((\|u\|,\|v\|)) n\lo 1\inv n\lo 0\inv (1-r\lo n)\udex {n\lo 1 -1 } (1-s\lo n) \udex {n\lo 0 - 1}, \nonumber
\end{eqnarray}
in which $\delta\lo w(x)$ is the Dirac function of $x$ such that it is zero whenever $x \neq w$ and $\tint \lo {\real\udex {dim(x)}} \delta\lo w (x) h(x) dx = h(w)$ for any function $h$ of $x$.
Similar to the above, we can show that $\cov(u\udex {\otimes {k}}, v\udex {\otimes l}) = O(n\udex {-1})$. By the compactness of $\Omega(X)$, this convergence is uniform on $(a, b)$, which means that $\cov\{(X\sam {J(i)} -X\sam i)\udex {\otimes k}, (X\sam {J(j)} - X\sam j)\udex {\otimes l}\} = O(n\udex{ - 1-2/p})$. Hence (\ref{eq:prf-thm-bias-1}) holds, which completes the proof.
\end{proof}

The theorem suggests that, compared to the variance, the bias of the matching estimator is more sensitive to the dimension of the balancing score. In fact, if we assume a higher-order differentiability of the outcome regression function in Assumption \ref{assumption-regularity}, then a higher-dimensional balancing score will be allowed for the resulting matching estimator to have variance of order $O(n\inv)$. In particular, we speculate that if the outcome regression function is smooth, then the variance of the matching estimator is always $O(n\inv)$, regardless of the dimension of the balancing score used. The simulation studies reveal this point.

From this theorem, whenever the dimension of the balancing score is greater than four, it determines the convergence order of the resulting estimator of $ACE$. Intuitively, this is caused by the fact that when the sample size is fixed and more components are added to the covariates, the empirical distribution of the covariates becomes more sparse, so that each subject $i$ is less similar to those in its neighborhood $J\lo m (i)$.
From this point of view, the propensity score is optimal among all, as it always generates a $\sqrt n$-consistent estimator. By contrast, using the original covariates is the worst choice.

\def\bopp{O\lo p\udex +}

When the central subspaces for both the treatment and control groups are at most four-dimensional, the reduced covariates from SDR become a comparable choice to the propensity score, as they also provide a $\sqrt n$-consistent estimator of $ACE$. This is very common in applications. For example, the central subspace is one-dimensional under the popular single-index model. In this special case, the following theorem suggests that the reduced covariates can further outperform the propensity score, in the sense that they produce an asymptotically stabler matching estimator. For regulation purpose, in addition to Assumption \ref{assumption-regularity}, we adopt the following assumption:

\begin{assumption}\label{assumption-regularity-complement}
For $t=0,1$, almost surely, $E\{Y\udex w (t)|X\}$ is Lipschitz continuous when $w=2$ and uniformly bounded when $w=4$, and $\var \{Y(t) | X\} > 0$.
\end{assumption}

\begin{thm}\label{thm:asym-var}
Suppose the ignorability assumption (\ref{assumption-ig}), Assumptions \ref{assumption-regularity}, \ref{assumption-cmmn-sppt-weak} and \ref{assumption-regularity-complement} hold. If
both $\csz$ and $\cso$ are one-dimensional, then
\begin{eqnarray*}
&& n \var (\hat\mu\lo r) \rightarrow V \lo X + \sum\lo {t=0}\udex 1 \left[ E \left\{\frac{\var \{Y(t) | X\}} {h(\pi (X'\beta\lo t),t)}\right\}\right. \nonumber \\
&& \hspace{0.5cm} + \left.\frac{1}{2m}E\left[\left\{\frac{1}{ h(\pi(X'\beta\lo t), t)} - h (\pi (X'\beta\lo t), t)\right\}\var \{Y(t) | X\}\right]\right],
\end{eqnarray*}
where $V\lo X = \var[E\{Y(1) - Y(0) | X \}]$, $h(\pi(X'\beta\lo 0), 0) = 1- \pi(X'\beta\lo 0)$ and $h(\pi(X'\beta\lo 1), 1) = \pi(X'\beta\lo 1)$. In addition, $\var (\hat\mu\lo r) \leq \var (\hat\mu\lo \pi)$ for all large $n$ when the latter exists, and the equality holds in one of the following cases:
\begin{itemize}
\item[(i)] $T \indep X$, and $E\{ Y(t) | X\} \equiv E\{Y(t)\}$ for $t=0, 1$.
\item[(ii)] $\csz = \cso$, $T \indep X | X'\beta\lo 0$, and $E\{Y(t) | X'\beta\lo t\} = E\{Y(t) | \pi (X'\beta\lo t)\}$ for $t=0, 1$.
\end{itemize}
\end{thm}

\begin{proof}
 The form of the variance follows directly from Theorem $5$ of \citet{abadie2006} and the sufficiency of $X'\beta\lo t$ for $Y(t) | X$. Let $\sigma\lo t\udex 2 (R) = \var \{Y(t) | R\}$ for any random element $R$.
To see that $\var (\hat \mu\lo \pi) \geq \var (\hat \mu\lo r)$ for all large $n$, following Theorem $5$ of \citet{abadie2006}, under Assumption \ref{assumption-cmmn-sppt-stronger}, we have
\begin{eqnarray*}
&& n \var (\hat\mu\lo \pi) \rightarrow V \lo {\pi(X)} + \sum\lo {t=0}\udex 1 \left[ E \left\{\frac{\sigma\lo t\udex 2 (\pi (X))} { h (\pi (X), t)}\right\}\right. \nonumber \\
&& \hspace{0.5cm} + \frac{1}{2m}\left.E\left[\left\{\frac{1}{ h(\pi(X),t)} - h (\pi (X), t)\right\} \sigma\lo t\udex 2 (\pi (X))\right] \right].
\end{eqnarray*}
We write $n \var (\hat \mu\lo \pi) = V\lo \pi + \one \lo \pi + \two\lo \pi$ and $n \var (\hat \mu\lo \pi) = V\lo X + \one \lo r + \two\lo r$, and additionally introduce $\one\lo X$ and $\two\lo X$, in which
{\small {
\begin{eqnarray*}
\one\lo \pi &=& E \left\{\frac{\sigma\lo 1\udex 2 (\pi(X))} {\pi (X)} +  \frac{\sigma\lo 0\udex 2 (\pi(X))} {1-\pi (X)} \right\} \\
\two\lo \pi &=& \frac{1}{2m} E\left[\left\{\frac{1}{\pi(X)} - \pi (X)\right\} \sigma\lo 1\udex 2 (\pi(X)) + \left\{\frac{1}{1-\pi(X)} - \{1-\pi (X)\} \right\} \sigma\lo 0 \udex 2 (\pi(X))\right] \\
\one\lo r &=& E \left\{\frac{\sigma\lo 1 \udex 2 (X)} {\pi (X'\beta\lo 1)} +  \frac{\sigma\lo 0 \udex 2 (X)} {1-\pi (X'\beta\lo 0)} \right\} \\
\two\lo r &=&  \frac{1}{2m}E\left[\left\{\frac{1}{\pi(X'\beta\lo 1)} - \pi (X'\beta\lo 1)\right\} \sigma\lo 1 \udex 2 (X)
+ \left\{\frac{1}{1-\pi(X'\beta\lo 0)} - \{1-\pi (X' \beta\lo 0)\} \right\} \sigma\lo 0 \udex 2 (X)\right]. \\
\one\lo X &=& E \left\{\frac{\sigma\lo 1\udex 2 (X)} {\pi (X)} +  \frac{\sigma\lo 0\udex 2 (X)} {1-\pi (X)} \right\} \\
\two\lo X &=&  \frac{1}{2m}E\left[\left\{\frac{1}{\pi(X)} - \pi (X)\right\} \sigma\lo 1 \udex 2 (X) + \left\{\frac{1}{1-\pi(X)} - \{1-\pi (X)\} \right\} \sigma\lo 0\udex 2 (X)\right].
\end{eqnarray*}}}
Then the inequality $n \var (\hat \mu\lo \pi) \geq n \var (\hat \mu\lo r)$ can be implied if
\begin{equation}\label{eq:prf-thm-asymvar}
V\lo \pi + \one \lo \pi \geq V\lo X + \one\lo X, \quad \two \lo \pi \geq \two\lo X, \quad \one\lo X \geq \one\lo r, \quad \two\lo X \geq \two\lo r.
\end{equation}
Let $\mu\lo {c,t}(X) = E\{Y (t) | X\} - E\{Y (t)  | \pi (X)\}$ for $t=0, 1$. By definition,
\begin{eqnarray}\label{eq:prf-thm-asymvar-1}
V\lo X - V\lo \pi &=& \var[E\{Y(1)- Y(0) | X\}] - \var[E\{Y(1) - Y(0) | \pi (X)\} ] \nonumber \\
&=& E\{\mu\lo {c,1} (X)\}\udex 2 + E\{\mu\lo {c,0} (X)\}\udex 2 - 2 E\{\mu\lo {c,1} (X) \mu\lo {c,0} (X) \}.
\end{eqnarray}
Let $\pi\udex *(X) = \pi(X) / \{1 - \pi(X)\}$, the logit function of $\pi(X)$. By the triangle inequality,
\begin{eqnarray*}
- 2 E\{\mu\lo {c,1}(X) \mu\lo {c,0}(X)\} &=& -2 E \left[ [\mu\lo {c,1}(X) \{\pi\udex *(X)\}\udex {-1/2} ] \, [ \mu\lo {c,0} (X) \{\pi\udex * (X)\} \udex {1/2}]\right] \\
&\leq& E\{\mu\lo {c,1} \udex 2 (X) / \pi\udex *(X)\} + E \{\mu\lo {c,0} \udex 2 (X) \pi\udex *(X)\}.
\end{eqnarray*}
By plugging it back to (\ref{eq:prf-thm-asymvar-1}), we have
\begin{equation*}
V\lo X - V\lo \pi \leq E[\mu\lo {c,1}(X)\udex 2 / \pi (X)] + E[\mu\lo {c,0} (X)\udex 2 / \{1- \pi (X)\}]
\end{equation*}
On the other hand, we have
\begin{eqnarray*}
\one\lo \pi - \one \lo X &=& E[\{\sigma\lo 1 \udex 2 (\pi(X)) - \sigma\lo 1\udex 2 (X)\} / \pi(X)] + E[\{\sigma\lo 0 \udex 2 (\pi(X)) - \sigma\lo 0\udex 2 (X)\} / \{1 - \pi(X)\}] \\
&=& E[\mu\lo {c,1} \udex 2 (X) / \pi (X)] + E[\mu\lo {c,0} \udex 2 (X) / \{1- \pi (X)\}].
\end{eqnarray*}
Hence $V\lo X - V\lo \pi \leq \one\lo \pi - \one\lo X$, which implies the first inequality in (\ref{eq:prf-thm-asymvar}). Because $E\{\sigma\lo 1 \udex 2 (X) | \pi(X)\}  = E[\var \{Y(1) | X\} | \pi(X)] \leq \var \{Y(1) | \pi(X)\} = \sigma\lo 1\udex 2 (\pi(X))$, we have
\begin{equation*}
E[\{\pi(X)\inv - \pi (X)\} \sigma\lo 1 (X) ] \leq E[\{\pi(X)\inv - \pi (X)\} \sigma\lo 1 (\pi(X)) ].
\end{equation*}
Similarly, we can show the corresponding inequality for the part when $t=0$ with $1- \pi (X)$ in place of $\pi(X)$. Thus $\two\lo \pi \geq \two\lo X$. Next, let $\phi: \real\udex + \rightarrow \real\udex +$ be that $\phi(x) = x\inv$. Then $\phi$ is a convex function. By Jensen's inequality, we have $\pi(X'\beta\lo 1) \inv = [E\{ \pi(X) | X'\beta\lo 1\}]\inv \leq E\{\pi\inv (X) | X'\beta\lo 1\}$, which, together with that $\sigma\lo 1\udex 2 (X) = \sigma\lo 1\udex 2 (X'\beta\lo 1)$, implies that
\begin{equation*}
E \{\sigma\lo 1 \udex 2 (X) / \pi (X'\beta\lo 1)\} \leq E [\sigma\lo 1 \udex 2 (X'\beta\lo 1) E\{\pi\inv (X) | X'\beta\lo 1\}] = E\{\sigma\lo 1\udex 2 (X) / \pi (X) \}.
\end{equation*}
Similarly, we can show the corresponding inequality for the part when $t=0$ with $1- \pi (X)$ in place of $\pi(X)$. Hence $\one\lo X \geq \one\lo r$. Finally, let $\psi: \real\udex + \rightarrow \real$ be that $\psi(x) = x\inv - x$, then $\psi$ is also a convex function. Thus the fourth inequality in (\ref{eq:prf-thm-asymvar}) can be shown similarly. From the arguments above, the equality between the limits of $n \var (\hat \mu\lo \pi)$ and $n \var (\hat \mu\lo r)$ holds if and only if all the four equalities in (\ref{eq:prf-thm-asymvar}) hold, which means that
\begin{itemize}
\item[(a)] $\mu\lo {c,1}(X) {\pi\udex *(X)}\udex {-1/2} = - \mu\lo {c,0} (X) {\pi\udex *(X)}\udex {1/2}$,
\item[(b)] $E\{Y(t) | X\} = E\{ Y(t) |\pi(X)\}$ for $t=0, 1$,
\item[(c)] $\pi(X) = \pi(X'\beta\lo 1) = \pi(X'\beta\lo 0)$.
\end{itemize}
Note that (b) is equivalent to that $\mu\lo {c,t} (X) = 0$ for $t=0, 1$, which implies (a).
If $\csz = \cso$, then (c) is equivalent to that $\pi(X) = \pi (X'\beta\lo 1)$. Otherwise, it means that $\pi(X)$ is a constant, which, together with (b), indicates that $E\{Y(t) | X\} \equiv E\{Y(t)\}$ for $t=0, 1$. This completes the proof.
\end{proof}

Case (i) in this theorem means that the covariates are redundant to both the propensity score and the outcome regressions. That is, the treatment is assigned completely at random, and both outcomes have constant mean over the subjects. However,  the covariates may affect the other aspects of the outcomes, such as the variance, so that the central subspaces do not vanish. Case (ii) in the theorem means that both the treatment assignment and the outcomes are affected by the same linear combination of the covariates, and furthermore, that the propensity score is sufficient for the outcome regressions.

Because both cases can be easily violated in practice, the theorem justifies that when the reduced covariates from SDR are univariate and known a priori, they are likely to outperform the propensity score, in terms of producing an asymptotically stabler estimator in the subsequent matching.

The underlying reason for the superiority of the reduced covariates from SDR, is their sufficiency for modeling the outcomes. In this sense, they are comparable to the prognostic score (\citealp{hansen2008}), defined as $E\{Y(0) | X\}$, which is also sufficient subject to the conditional independence between $Y(1)$ and $X$ given $Y(0)$.
Following the arguments in Theorem \ref{thm:asym-bias} and Theorem \ref{thm:asym-var}, this prognostic score can further outperform the proposed reduced covariates in reducing both the bias and the variance of the matching estimator. However, the additional assumption for the prognostic score involves the joint distribution of both potential outcomes, and its estimation usually relies on parametric models, neither of which can be readily justified in practice.

The reduced covariates from SDR are also comparable to the aforementioned active set in \citet{de2011}, in the sense that they both  pre-assume a sparsity structure in the data.
As mentioned before, the cardinality of the active set is always greater than or equal to the dimensions of the central subspaces, which, by Theorem \ref{thm:asym-bias}, makes the subsequent matching potentially less effective. 

To further reduce the size, \citet{de2011} used the above-mentioned active set in place of the original covariates to further conduct variable selection in modeling the treatment assignment, and to remove the redundant covariates. The cardinality of the reduced set is possibly smaller than the dimensions of the central subspaces. However, as the reduced set is no longer sufficient for modeling the outcomes, Theorem \ref{thm:asym-var} suggests that it may enlarge the  variance of the subsequent matching estimator. Nonetheless, this procedure motivates us to develop a sequential SDR procedure, which can be useful when there is need to further reduce the working central subspaces. For continuity of context, we leave more details to the Discussion.

\section{Sample estimation and implementation}

We now present the details in using SIR to estimate the central subspaces, and then using the resulting reduced covariates in matching to estimate the causal effects.

Following the notations in Section 2,
for standardized covariates with zero mean and identity covariance matrix, SIR estimates the central subspace through an intermediate parameter $M = E\{E(X |W) E(X' | W)\}$, called the candidate matrix. Given a matrix estimate $\hat M$, the central subspace is estimated as the linear space spanned by the leading eigenvectors of $\hat M$. Because no model is assumed on $E(X | W)$, to construct $\hat M$, \citet{li1991} proposed the slicing strategy, which divides the sample into $H$ slices according to the similarity in $W$. $E(X|W)$ is then estimated as the sample mean of $X$ within each slice, and $\hat M$ is given by the usual sample moments.

Given the reduced covariates, we are free to use any existing matching method. As an illustration, here we use the Abadie-Imbens estimator with the Mahalanobis distance (\citealp{rosenbaum1985}) to measure the similarity of subjects.
For a random vector $Z$ with the sample values $\{Z\sam i, i=1,\ldots,n\}$, let $\hat \Sigma\lo Z$ be its sample covariance matrix. The Mahalanobis distance between $Z\sam i$ and $Z\sam j$ is
\begin{equation}\label{eq:mahalanobis}
D(Z\sam i, Z\sam j)=\{(Z\sam i - Z \sam j)' \hat \Sigma\lo Z\inv (Z\sam i - Z \sam j)\}\udex {1/2}.
\end{equation}

When $Z$ is used for matching, for each subject $i$, the matched set $J\lo m (i)$ in (\ref{eq:impute}) contains the $m$ subjects in the other treatment group whose $Z$ values are closest to $Z\sam i$ in terms of Mahalanobis distance. After calculating the imputed values by (\ref{eq:impute}), we estimate $ACE$ by

{\small {
\begin{equation}\label{eq:est-ace}
\widehat{ACE}=\left[\tsum_{T\sam i = 0} \{\hat {Y}\sam i (1)- {Y}\sam i (0)\} + \tsum_{T\sam i = 1} \{{Y}\sam i (1)- \hat Y\sam i (0)\}\right] / n.
\end{equation}}}

In summary, the proposed matching approach can be implemented in the following steps:

\begin{enumerate}
\item [Step 1.] SDR: estimate the reduced covariates in the control group as follows.
 \begin{itemize}
 \item Standardize $X$ to be $\hat \Sigma\lo {X | T=0} \udex {-1/2} \{X - \hat E (X| T=0)\}$, where $\hat E (X | T=0)$ and $\hat \Sigma\lo {X|T=0}$ are the sample mean and the sample covariance matrix of $X$ in the control group, respectively. For simplicity, we still denote the standardized variable by $X$, if no ambiguity is caused.
 \item For $j=1,\ldots, H$, let $q\lo j$ be the $(j/H)$th sample quantile of $\{Y \sam i (0): T\sam i = 0\}$. Construct the $j$th slice as $ S\lo j (Y(0)) = I (q\lo {j-1} < Y (0) \leq q\lo {j})$, where $I(\cdot)$ is the indicator function and $q\lo 0$ is a constant that is less than any observed $Y\sam i (0)$. Estimate $E(X | S\lo j (Y(0))=1, T\sam i = 0)$ by
 \begin{equation*}
 \hat \mu\lo {X, j} = \tsum_ {T^ {(i)} = 0} S\lo j (Y\sam i (0) ) X\sam i \, / \, \tsum _ {T^ {(i)} = 0} S\lo j (Y\sam i (0)).
 \end{equation*}
 \item Estimate $M\lo 0= E[ E\{X | Y(0), T=0\} E\{X' | Y(0), T = 0\}]$ by $\hat M\lo 0 = \tsum\lo {j=1}\udex H \hat \mu\lo {X, j} \hat \mu\lo {X,j}' / H$, and use the sequential tests to determine the rank of $M\lo 0$. Let $\hat r (0)$ be the selected rank and $\hat \beta\lo 0$ be the $\hat r(0)$ eigenvectors of $\hat M\lo 0$ with the largest eigenvalues.
\end{itemize}

\item [Step 2.] Match each subject in the treatment group with those in the control group: for subject $i$, let ${\mathcal D}\lo i = \{D(X\udex {(i) \prime}\hat \beta\lo 0, X\udex {(j)\prime} \hat \beta\lo 0): T\sam j =0\}$, where $D$ is the Mahalanobis distance in (\ref{eq:mahalanobis}), and $J\lo m (i)$ be the set that corresponds to the $m$ smallest values in ${\mathcal D}\lo i$. Use (\ref{eq:impute}) to compute $\hat Y\sam i (0)$.

\item [Step 3.] Conduct Step 1 for the treatment group, and conduct Step 2 to compute $\hat Y\sam i (1)$ for each subject $i$ in the control group.

\item [Step 4.] Estimate $ACE$ by (\ref{eq:est-ace}).
\end{enumerate}

In practice, the number of slices $H$ is usually fixed at a small constant, such as $5$ or $10$. Although it imposes the constraint that the reduced covariates derived from SIR are at most $(H-1)$-dimensional, the issue can be easily addressed: if one finds that the estimated reduced dimension $\hat r(0)$ or $\hat r (1)$ coincides with $H-1$, then one can raise the value of $H$ to see whether more directions in the central subspace can be recovered.

When the causal effect of interest is $ACET$, the implementation can be adjusted by skipping step 3 and using the results in step 2 to compute
\begin{equation*}
\widehat{ACET}= \tsum\lo {i=1}\udex n T\sam i \{Y\sam i  (1) - \hat {Y}\sam i (0)\} / \tsum\lo {i=1}\udex n T\sam i.
\end{equation*}

\section{Simulation Studies}

In this section, we illustrate the effectiveness of SDR for matching, by applying the reduced covariates to the Abadie-Imbens estimator. For comparison, we also apply the estimated propensity score derived from CBPS \citep{imai2014}, the true propensity score, the original covariates, and the active set of covariates \citep{de2011} to the same estimator, respectively. The active set is assumed  known a priori. To assess the effect of SDR estimation in the proposed method, we also apply the reduced covariates induced by the true central subspaces, which is oracle, just like the active set of covariates. In addition, we include genetic matching on the original covariates in the comparison, for the reason that the method has the potential to detect and use the sparsity structure in the data.

Because the linearity condition (\ref{assumption-linearity}) required for the consistency of SIR in each treatment group can be violated in practice, we do not treat the condition as granted. Instead, we let the covariates be elliptically distributed conditional on the treatment assignment, elliptically distributed marginally, and  composed of a mixture of continuous and discrete components, sequentially in this section.

\subsection{Case 1: Elliptical covariates conditional on treatment}

In the first scenario, we generate $T$ marginally from a Bernoulli distribution with mean $0.5$, and then generate $X | T=t$ from $N(\mu\lo t, \Sigma\lo t)$, where the $(i,j)$th entry of $\Sigma\lo t$ is $\delta\lo t \udex {|i - j|}$. We set $\mu\lo 0$ at the origin of $\real\udex p$ and $\delta\lo 0$ at $0.2$, and let $\mu_1$ and $\delta_1$ vary with the models. When $\delta\lo 1 = 0.2$, the distribution $X|T=t$ depends on $t$ only through the conditional mean $\mu_t$, which is the case discussed in \citet{rubin1992}. The true propensity scores can be easily calculated using Bayes' Theorem.

We set $n=500$, $p=10$, and study the following four models:
\begin{itemize}
\item[\one.] $Y (t) = (t + 2) (X\lo 1 + 1.5)\udex 2 + t + \varepsilon$, $\mu\lo {1,k} = \mu\lo 0$, $\delta\lo 1 = 0.5$;
\item[\two.] $Y (t) = 2 \sin\{0.5 (X\lo 1 + X\lo 2 - X\lo 3)\} + t + \varepsilon$, $\mu\lo {1,k} = I(k \leq 3) / 3$, $\delta\lo 1 = 0.5$;
\item[\three.] $Y (t) = X\lo 1 + X\lo 2 + X\lo 3 + t(X\lo 4 + X\lo 5) + \varepsilon$, $\mu\lo {1,k} \equiv p\udex {-1/2}$, $\delta\lo 1 = 0.2$;
\item[\four.] $Y (t) = 3 \sin \{X'\Sigma\lo 1\inv (\mu\lo 1 - \mu\lo 0) / 3\} + t \, \{X'\Sigma\lo 1\inv (\mu\lo 1 - \mu\lo 0)\}\udex 2 / 3 + \varepsilon$, $\mu\lo {1,k} = c\lo 1 k$, $\delta\lo 1 = 0.2$.
\end{itemize}
Here $\epsilon$ is an independent error generated from $N(0, 0.5\udex 2)$, $\mu\lo {1,k}$ is the $k$th component of $\mu\lo 1$ for $k=1,\ldots, p$, and $c\lo 1$ is the constant that makes $\|\mu\lo 1\|\lo 2 = 0.5$ in Model \four. Obviously, the ignorability assumption (\ref{assumption-ig}) is satisfied in all the models.

In Model \two, the causal effects are constant among all the subjects. In Models \one \ and \two, $\delta\lo 1$ differs from $\delta\lo 0$, so that the correlation structure of the covariates varies between treatment groups. From the view of SDR, in Model \three,  the two central subspaces differ,  and in all the models the central subspaces are one-dimensional. Thus, by Theorem \ref{thm:asym-var}, the reduced covariate from oracle SDR should outperform the true propensity score in Models \one, \two\ and \three, in the sense that the induced Abadie-Imbens estimator has a smaller asymptotic variance. In Model \four, the true propensity score is an invertible function of $X'\Sigma\lo 1\inv (\mu\lo 1 - \mu\lo 0)$, which is also the reduced covariate from oracle SDR. Therefore, the model falls into case (ii) of Theorem \ref{thm:asym-var} indicating that  the two balancing scores are equally good for matching asymptotically.

The strong common support condition is satisfied at the population level in all four models. However, because the dimension $p$ is moderately large, the covariates have a relatively sparse support at the sample level, which makes the actual realization of the condition questionable. In contrast, because the reduced covariate from SDR is univariate in all the models, the issue is less severe. In fact, depending on how different the central subspaces are from $\mu\lo 1 - \mu\lo 0$, $\Sigma\lo 0$ and $\Sigma\lo 1$, the weak common support condition is realized almost perfectly in Models \one \ and \two \ and to a reasonable extent in the other models. 


We set $m=1$, and estimate both $ACE$ and $ACET$ in each model. The results for $ACE$ are summarized in Table \ref{table:case1-ace} and the results for $ACET$ are displayed in Table 1 in the Appendix B (Supplementary Material), respectively, which include the evaluation of the bias and the standard deviation of each estimator, based on $1000$ independent runs. For readers' convenience, we also record the root mean square error (RMSE) of each estimator, which can be calculated from the bias and the standard deviation, subject to rounding errors.

\begin{table}
\center
\begin{threeparttable}
\caption{\label{table:case1-ace} Simulation Results for $ACE$ in Case 1}
\begin{tabular}{ccccc}								
		& \one	&	\two	&	\three	&	\four			\\
\hline
&\multicolumn{4}{c}{BIAS}\\						
Ambient	&	-0.1460	&	0.0929	&	0.4764	&	0.0571	\\
Estimated PS	&	-0.0107	&	0.0317	&	0.0415	&	0.0154\\
True PS	&	0.0102	&	0.0073	&	0.0211	&	0.0013	\\
Genetic Matching & -0.1572&	0.0228&	0.3116&	0.0222\\
Active Set (Oracle)	&	-0.0138	&	0.0172	&	0.2005	&	0.0571	\\
SDR (Oracle)	&	-0.0138	&	0.0006	&	0.0164	&	0.0014	\\
Proposed	&	-0.0065	&	0.0029	&	0.0220	&	0.0016	\\
&\multicolumn{4}{c}{SD}\\						
Ambient	&	0.4258	&	0.0849	&	0.1388	&	0.0549	\\
Estimated PS	&	0.6463	&	0.1051	&	0.1751	&	0.0567	\\
True PS	&	1.1884	&	0.1856	&	0.2119	&	0.0537	\\
Genetic Matching & 0.3308&	0.0637&	0.1283	&0.0554\\
Active Set (Oracle)	&	0.1583	&	0.0559	&	0.1054	&	0.0548	\\
SDR (Oracle)	&	0.1583	&	0.0535	&	0.0917	&	0.0537	\\
Proposed	&	0.1788	&	0.0551	&	0.1011	&	0.0525	\\
&\multicolumn{4}{c}{RMSE}\\													
Ambient	&	0.4499	&	0.1258	&	0.4962	&	0.0792	\\
Estimated PS	&	0.6460	&	0.1097	&	0.1798	&	0.0588	\\
True PS	&	1.1878	&	0.1856	&	0.2129	&	0.0537	\\
Genetic Matching & 0.3661&	0.0677&	0.3370	& 0.0597\\
Active Set (Oracle)	&	0.1588	&	0.0585	&	0.2265	&	0.0792	\\
SDR (Oracle)	&	0.1588	&	0.0535	&	0.0931	&	0.0537	\\
Proposed	&	0.1789	&	0.0552	&	0.1034	&	0.0525	\\
\hline
\end{tabular}
\begin{tablenotes}
\small \item Here, ``Ambient" refers to the Abadie-Imbens estimator based on the Mahalonobis distance calculated from the original covariates;`` Estimated PS" is the same estimator based on the estimated propensity scores from CBPS; ``True PS" is based on the true propensity score; ``Active Set (Oracle)" is based on  the Mahalonobis distance calculated from the true active set of covariates; ``SDR (Oracle)" is the proposed approach with true reduced covariates; and  ``Proposed" is the proposed approach with reduced covariates estimated from SIR.
\end{tablenotes}
\end{threeparttable}
\end{table}

From these tables, the reduced covariates from SDR estimation substantially outperform the original covariates for matching, in all the aspects in Models \one \ to \three\ and in terms of reduced bias in Model \four. They are comparable with the active set of covariates in Models \one \ and \two, and superior to the latter on bias reduction in Models \three\ and \four. The improvement is due to the relaxation of the ``curse of dimensionality", for the reason that all these three balancing scores are sufficient for the outcome regressions, and that the cardinality of the active set is small in Models \one\ and \two, and is relatively large in Models \three\ and \four. The genetic matching is generally superior to the Abadie-Imbens estimator based on the original covariates, but is inferior to the same estimator based on the reduced covariates from SDR estimation. This is a sign that at the current sample size, the sparsity structure of the data can be detected partially, but not completely, by genetic matching.

Interestingly, the true propensity score is slightly inferior to its estimator in terms of a larger variance in the resulting matching. The phenomenon has also been noticed by multiple authors, see \cite{lunceford2004}. Compared to the estimated propensity score, the reduced covariates from SDR estimation makes matching more stable in Models \one, \two, and \three, and equally consistent in Model \four, which complies with the theoretical anticipation from Theorem \ref{thm:asym-var}. Compared to the oracle SDR, the SDR estimation does cost a larger variation in matching, although the price is nearly negligible.

\subsection{Case 2: Elliptical covariates in the merged sample}

When applying SIR to data with continuous but non-elliptically distributed covariates, it has been a common practice to transform each covariate to be univariate normally distributed, and assume that the transformed covariates have a jointly normal distribution. In causal inference, if the researcher believes that the covariates affect the outcome in similar patterns before and after the treatment assignment,  it will be desired to transform the covariates uniformly across the treatment groups, which corresponds to an elliptically distributed $X$ but not necessarily elliptically distributed $X|T$. Consequently, the linearity condition (\ref{assumption-linearity}) on $X|T$ can be violated. 

To study the performance of the proposed approach in this case, we still set $n=500$, $p=10$, and generate $X$ from $N({\mathbf 0}, I\lo p)$. Because the outcome regressions are similar in different treatment groups in Models \one \ and \two, given $X$, we generate the outcomes from these two models, and use the corresponding propensity score to generate $T$ from an independent Bernoulli distribution. The new models satisfy the ignorability assumption (\ref{assumption-ig}), and differ from their counterparts in the previous subsection only in the marginal distribution $X$. We label them as Model \one$\udex *$ and Model \two$\udex *$, respectively. The results from $1000$ independent runs are summarized in Table \ref{table:case2-acet}.

\begin{table}
\center
\caption{\label{table:case2-acet} Simulation Results for $ACE$ and $ACET$ in Case 2}
\begin{tabular}{ccccc}	
&\multicolumn{2}{c}{$ACE$}&\multicolumn{2}{c}{$ACET$}\\	
\hline						
		&I$\udex *$	&	\two$\udex *$	&	I$\udex *$	&	\two$\udex *$		\\
\hline
&\multicolumn{4}{c}{BIAS}\\						
Ambient	&		-0.6805	&	0.0279	& 0.3132	&	0.0293	\\
Estimated PS		&	-0.4273	&	-0.0609	&	-0.3044	&	0.0059\\
True PS		&	-0.0209	&	-0.0072	&	0.0166	&	0.0047\\
Genetic Matching&	-0.6165&	0.0125& 0.0765&	0.0148\\
Active Set (Oracle)		&	-0.0961	&	0.0036	&	0.0113	&	0.0056\\
SDR (Oracle)		&	-0.0961	&	-0.0003	&	0.0113	&	0.0002\\
Proposed		&	-0.0761	&	0.0021	&	0.0187	&	0.0032\\
&\multicolumn{4}{c}{SD}\\						
Ambient		&	0.5151	&	0.0957	&	0.5796	&	0.1109\\
Estimated PS		&	0.8401	&	0.1312	&	0.7713	&	0.1419\\
True PS		&	2.1383	&	0.2997	&	1.1095	&	0.2066\\
Genetic Matching&	0.4368&	0.0763& 0.4526&	0.0852\\
Active Set (Oracle)		&	0.1822	&	0.0645	&	0.2712	&	0.0732\\
SDR (Oracle)		&	0.1822	&	0.0615	&	0.2712	&	0.0669\\
Proposed		&	0.2202	&	0.0653	&	0.2896	&	0.0729\\
&\multicolumn{4}{c}{RMSE}\\													
Ambient		&	0.8533	&	0.0996	&	0.6586	&	0.1146\\
Estimated PS		&	0.9421	&	0.1446&	0.8288	&	0.1419	\\
True PS		&	2.1373	&	0.2996	&	1.1091	&	0.2066\\
Genetic Matching &	0.6572&	0.0623& 0.3596&	0.0681\\
Active Set (Oracle)		&	0.2059	&	0.0645 &	0.2713	&	0.0734	\\
SDR (Oracle)	&	0.2059	&	0.0615		&	0.2713	&	0.0669\\
Proposed		&	0.2329	&	0.0653	&	0.2901	&	0.0729\\
\hline
\end{tabular}
\end{table}

The results are similar to the previous subsection. The reduced covariates from SDR estimation still outperform the original covariates and the propensity score, no matter whether the latter is estimated or known a priori. They are generally comparable to the reduced covariates from oracle SDR and the active set of covariates, indicating that the price of estimating the central subspaces is nearly negligible, and that the dimensionality of the balancing score does not harm the accuracy of matching when it is less than four, as expected in Theorem \ref{thm:asym-bias}.

\subsection{Case 3: Covariates that contain discrete components}

As mentioned in Section $2$, when $X$ is non-elliptically distributed, the linearity condition can be approximately satisfied, as long as the dimension $p$ is moderately large. Consequently,
the proposed matching approach can still be reasonably effective in finite samples. In this subsection, we examine its performance when $X$ is a mixture of elliptically distributed covariates and binary covariates.

The simulation setup follows \cite{lee2010}.  First, we generate ten covariates $(X\lo 1,\ldots, X\lo {10})$, in which $(X\lo 1, \ldots, X\lo 3)$ are only associated with $Y$, $(X\lo 4, \ldots, X\lo 6)$ are only associated with $T$, and $(X\lo 7,\ldots, X\lo {10}$) are associated with both $Y$ and $T$. Six covariates, $(X\lo 1, X\lo 3, X\lo 5, X\lo 6, X\lo 8, X\lo 9)$, have Bernoulli distribution with mean 0.5 marginally, and the others follow standardized normal distribution marginally. The correlation within each pair of covariates is zero, except that:
\begin{equation*}
corr(X\lo 1, X\lo 5)=corr(X\lo 3, X\lo 8)=0.2, \,\,corr (X\lo 2, X\lo 6) = corr (X\lo 4, X\lo 9)=0.9.
\end{equation*}
$T$ is generated from the conditional distribution $T|X$, in which the propensity score satisfies:
\begin{equation}\label{logit}
\logit\{\pi (X)\}= \alpha ' g (X).
\end{equation}
The observed outcome $Y (T)$ is generated from the linear model
\begin{equation*}
Y (T) =\omega 'X -0.4  T + \epsilon,
\end{equation*}
where $\epsilon$ follows $N(0,0.1^2)$. The function $g$ and the coefficient vectors $\alpha$ and $\omega$ are specified in the each of scenarios listed in the following. These scenarios differ in the degree of linearity and additivity in the propensity score (\ref{logit}). For more details, see \cite{lee2010}.
\begin{description}
\item[A:] Additivity and linearity (main effects only);
\item[B:] Mild non-linearity (one quadratic term);
\item[C:] Moderate non-linearity (three quadratic terms);
\item[D:] Mild non-additivity (four two-way interaction terms);
\item[E:] Mild non-additivity and non-linearity (three two way interaction terms and one quadratic term);
\item[F:] Moderate non-additivity (ten two-way interaction terms);
\item[G:] Moderate non-additivity and non-linearity (ten two-way interaction terms and three quadratic terms).
\end{description}

Again, $1000$ independent samples are generated, and the Abadie-Imbens estimators based on different balancing scores are used to estimate the causal effects. The results for $ACE$ are summarized in Table \ref{table:case3-ace} and the results for $ACET$ are displayed in Table 2 in the Appendix B (Supplementary Material). Same as before, the proposed approach yields smaller
bias and variance than the existing methods in all the scenarios, and is only slightly worse than that based on the oracle SDR.

\begin{table}
\caption{\label{table:case3-ace} Simulation Results for $ACE$ in Case 3}

\begin{tabular}{cccccccc}
&		A	&	B	&	C	&	D	&	E	&	F	&	G	\\
\hline
&\multicolumn{7}{c}{BIAS}\\						
Ambient	&	0.0551	&	0.0527	&	0.0508	&	0.0625	&	0.0586	&	0.0654	&	0.0570	\\
Estimated PS	&	0.0075	&	0.0041	&	0.0021	&	0.0073	&	0.0124	&	0.0097	&	0.0149	\\
True PS	&	0.0038	&	0.0020	&	0.0078	&	0.0046	&	0.0064	&	0.0041	&	0.0070	\\
Genetic Matching& 0.0343&0.0336&0.0246&0.0371&0.0353&0.0371&0.0321\\
Active Set (Oracle)	&	0.0247	&	0.0240	&	0.0208	&	0.0295	&	0.0287	&	0.0328	&	0.0247	\\
SDR (Oracle)	&	0.0023	&	0.0019	&	0.0025	&	0.0027	&	0.0025	&	0.0025	&	0.0023	\\
Proposed	&	0.0035	&	0.0037	&	0.0043	&	0.0050	&	0.0049	&	0.0055	&	0.0034	\\
&\multicolumn{7}{c}{SD}\\	
Ambient	&	0.0420	&	0.0431	&	0.0427	&	0.0445	&	0.0447	&	0.0440	&	0.0450	\\
Estimated PS	&	0.0692	&	0.0721	&	0.0661	&	0.0757	&	0.0745	&	0.0738	&	0.0731	\\
True PS &	0.0943	&	0.0950	&	0.1093	&	0.0983	&	0.1015	&	0.1004	&	0.1177	\\
Genetic Matching&0.0596&0.0590&0.0558&0.0600&0.0619&0.0574&0.0521\\
Active Set (Oracle)	&	0.0587	&	0.0593	&	0.0587	&	0.0601	&	0.0605	&	0.0625	&	0.0636	\\
SDR (Oracle)	&	0.0327	&	0.0335	&	0.0329	&	0.0331	&	0.0337	&	0.0334	&	0.0339	\\
Proposed	&	0.0365	&	0.0355	&	0.0356	&	0.0374	&	0.0372	&	0.0369	&	0.0357	\\
&\multicolumn{7}{c}{RMSE}\\		
Ambient	&	0.0693	&	0.0681	&	0.0663	&	0.0767	&	0.0737	&	0.0788	&	0.0726	\\
Estimated PS	&	0.0695	&	0.0722	&	0.0661	&	0.0761	&	0.0755	&	0.0744	&	0.0745	\\
True PS	&	0.0943	&	0.0950	&	0.1096	&	0.0984	&	0.1016	&	0.1004	&	0.1178	\\
Genetic Matching&0.0687&0.0679&0.0609&0.0705&0.0712&0.0683&0.0612\\
Active Set (Oracle)	&	0.0637	&	0.0640	&	0.0623	&	0.0669	&	0.0670	&	0.0705	&	0.0682	\\
SDR (Oracle)	&	0.0328	&	0.0336	&	0.0329	&	0.0331	&	0.0338	&	0.0334	&	0.0340	\\
Proposed	&	0.0366	&	0.0356	&	0.0359	&	0.0377	&	0.0375	&	0.0373	&	0.0358	\\
\hline
\end{tabular}
\end{table}

\section{Data Application}

We now illustrate the proposed methodology using a well-known dataset: LaLonde dataset. This dataset has been analyzed by \cite{lalonde1986} and \cite{dehejia1999} to evaluate the causal effect of a labor training program called National Supported Work (NSW) Demonstration on  earnings for job-seekers, who had economic and social problems before the enrollment in the program. In this study, the outcome is the individuals' earnings in 1978 (Ee78). The treatment variable is the indicator for the enrollment in the labor training program (Treat). The ten potential confounders are: age (Age), years of schooling (Educ), indicator for Blacks (Black), indicator for Hispanics (Hisp),  indicator for being married (Married), indicator for high school diploma (Nodegr), real earnings in 1974 (Re74), real earnings in 1975 (Re75), indicator variable for earnings in 1974 being zero (U74), and indicator variable for earnings in 1975 being zero (U75). Here we consider the subset of the original dataset referred as ``CPS-3" by  Dehejia and Wahba (1999), which contains 185 subjects who participated in NSW and 429 controls who  did not participate in NSW.

Since the focus  is to examine the impact of the labor training program on postintervention earnings for those who are eligible for the  program, the parameter of interest is $ACET$. \cite{dehejia1999} suggested a one-to-one matching  approach  with replacement, based on the estimated propensity score from parametric models. To apply the proposed matching approach to the dataset, we apply SIR in the control group to obtain the reduced covariates. The sequential  test for SIR (see the Appendix A) indicates that the  dimension of the reduced covariates is two. A scatterplot between each  reduced covariate and the outcome in the control group is shown in Figure \ref{dr}. From these plots, both reduced covariates clearly affect the outcome marginally, suggesting that two is a reasonable choice for the reduced dimension.

To check  the quality of matching, we draw the boxplot for each reduced covariate, grouped by the treatment assignment. For comparison, we also draw the boxplots for the estimated propensity scores from the three parametric models (See Appendix C in the Supplementary Material)  used in \cite{dehejia1999}. From the upper panel of  Figure~\ref{common},  when any of the estimated propensity scores is used,  the majority of the treated subjects will be matched with a small number of extreme outliers in the control group, which result in  loss of power in the subsequent causal effect estimation. Besides, the sparsity of the empirical support of the outliers will also cause non-negligible difference within the matched pairs, and lead to large bias in the estimated casual effect. By contrast, from the lower panel of Figure~\ref{common}, the issue is much less severe when the reduced covariates from SIR are used for matching.

\begin{figure}
\center
\includegraphics[scale=0.5, width=8cm]{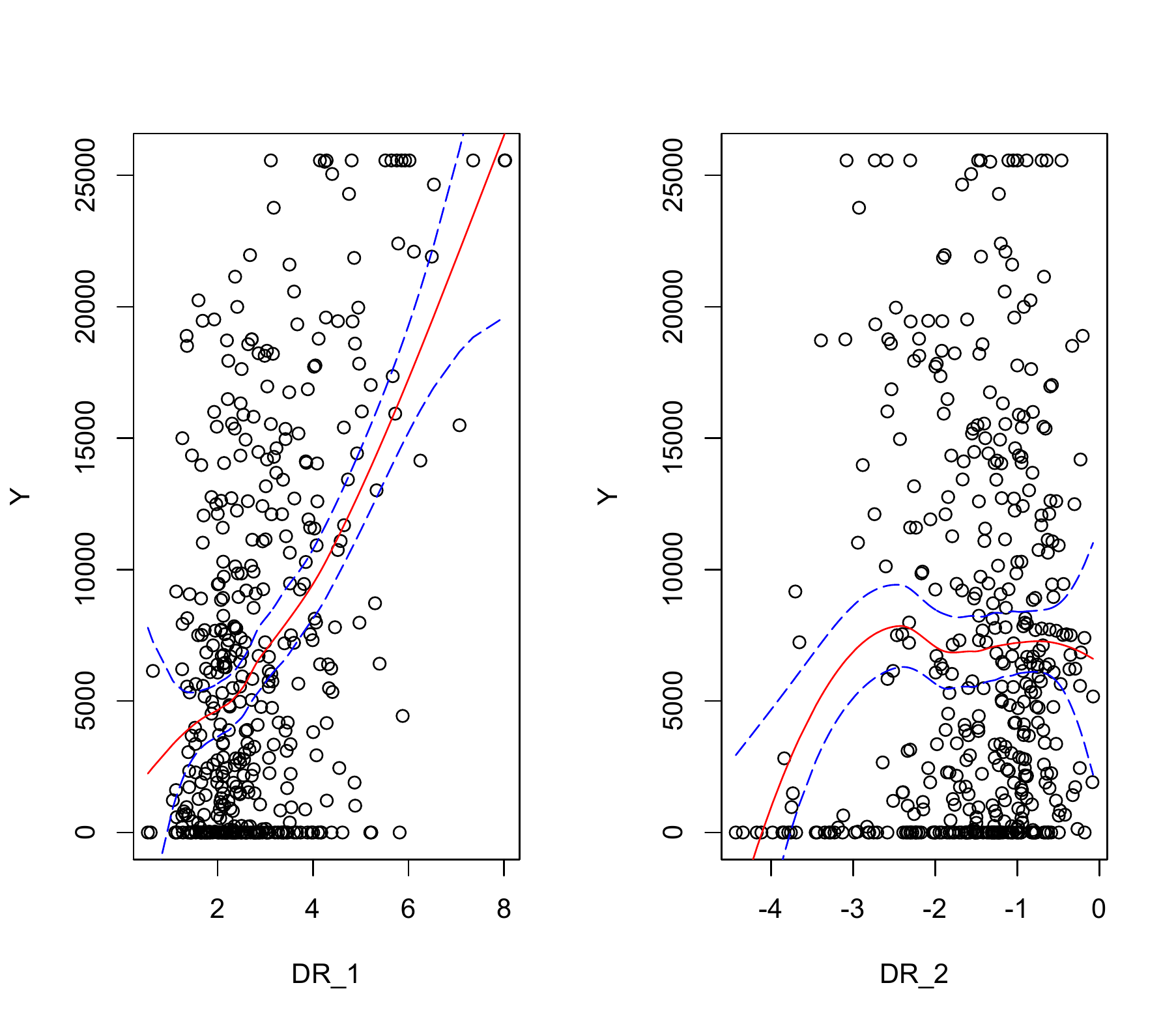}
\caption{\label{dr} Scatterplots of the response variable versus the first two reduced covariates in the control group.}
\end{figure}

\begin{figure}[h!]
\center
\includegraphics[scale=0.6]{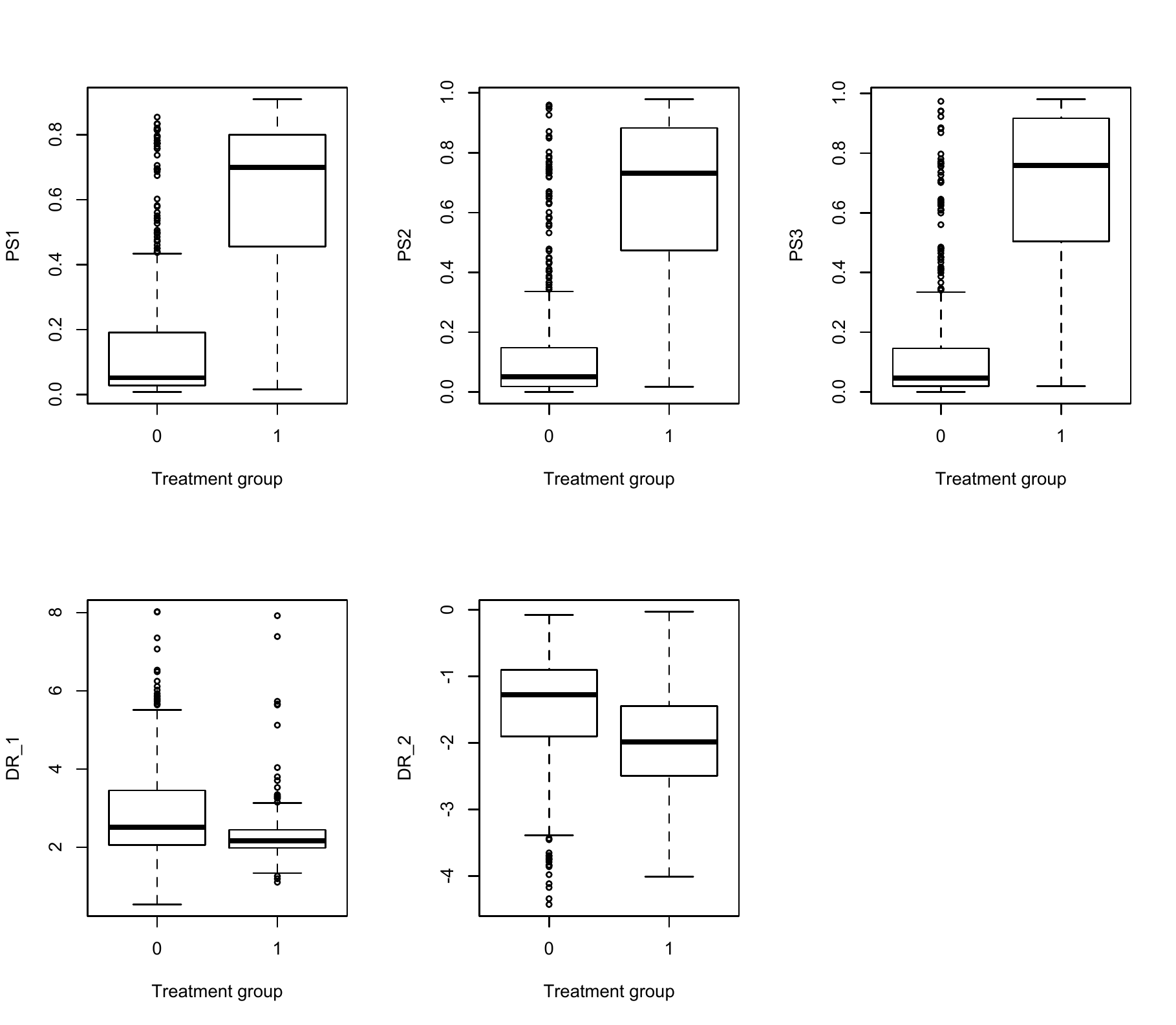}
\caption{\label{common} Boxplot of the propensity scores/reduced covariates between the treatment group and the control group. The upper panel is the boxplot of the estimated propensity scores from three different propensity score models (PS1--PS3 displayed in the Appendix C of the Supplementary Material) and the lower panel is the boxplot of the first two reduced covariates (DR\_1 and DR\_2) obtained from SIR. }
\end{figure}

We further plot the histogram of the two reduced covariates (Figure~\ref{common3}) and the estimated propensity scores (Figure~\ref{common2})  between the treatment and the control group.  For propensity score-based approaches, there are a few bins in which the control group is much smaller than the treatment group, which means it will be hard to find matches for the treated subjects in those bins. However, for the proposed approach, the control group is almost always larger than  the treatment group, which indicates there is enough overlapping for the estimation of the causal effect to be reliable.

 The proposed method estimates $ACET$ to be 205 suggesting that the job-seeking program increases the individual earnings at 1978.  This is consistent with the conclusion from an experimental study in \citet{dehejia1999}. In contrast, if we apply the Abadie-Imbens estimator based on the original covariates or the estimated propensity scores, the Abadie-Imbens estimator estimates $ACET$ to be -361 (original covariates), -835 (PS1), -528 (PS2), -304 (PS3), respectively, which leads to the opposite conclusion.

\begin{figure}[h!]
\center
\includegraphics[scale=0.6]{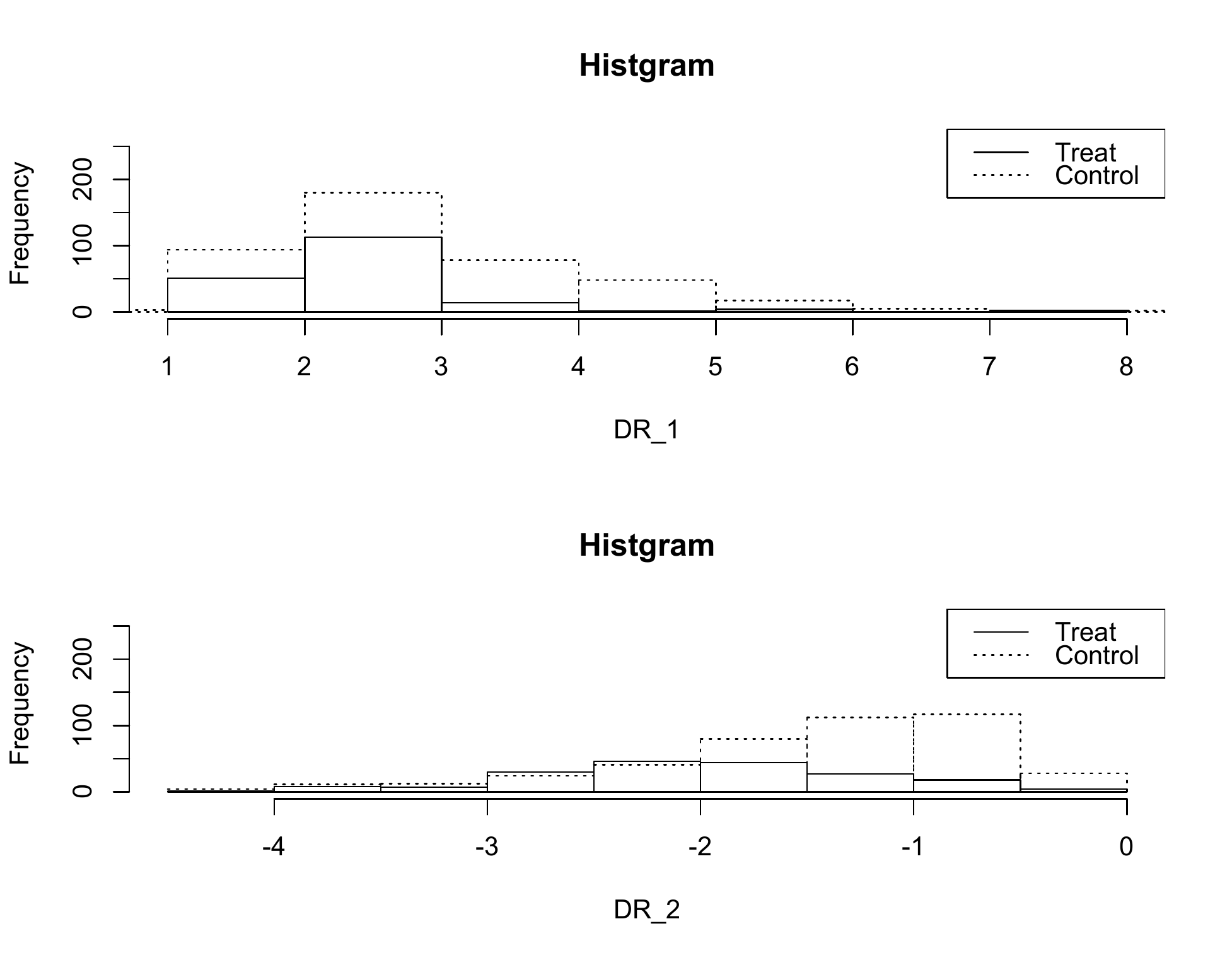}
\caption{\label{common3} Histograms of the two reduced covariates (DR\_1 and DR\_2) between the treatment group and the control group.}
\end{figure}

\begin{figure}[h!]
\center
\includegraphics[scale=0.6]{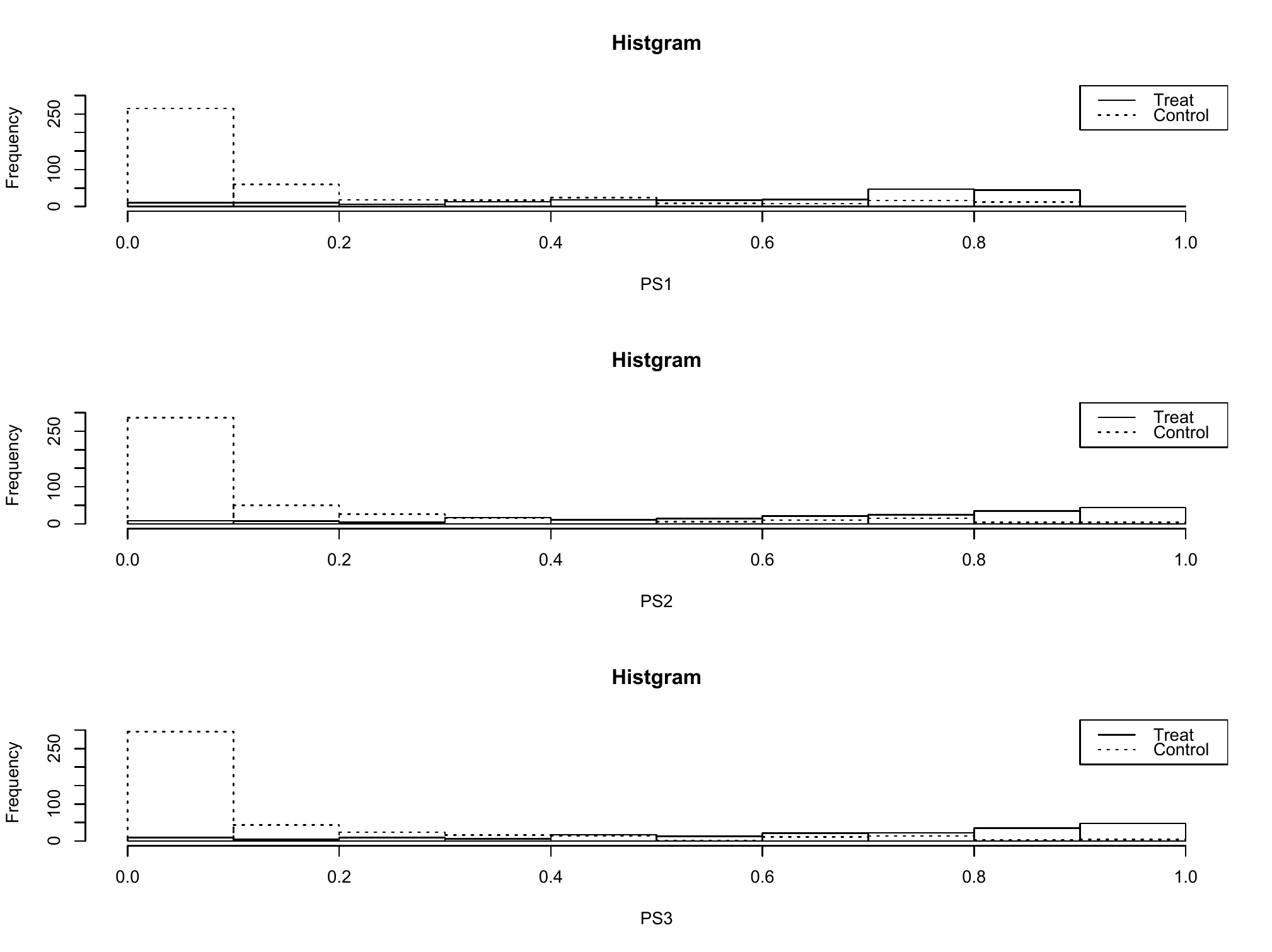}
\caption{\label{common2} Histograms of the estimated propensity scores (PS1, PS2 and PS3) between the treatment group and the control group.}
\end{figure}


\section{Discussion}
When the dimension of the proposed reduced covariates in either treatment group is greater than four, as Theorem \ref{thm:asym-bias} shows, it is desirable to further reduce the covariates to enhance the estimation accuracy of matching.
Following the discussion in Section $4$, we can develop sequential SDR by using the outcome and the treatment assignment alternatively as the response variable in dimension reduction. The procedure gives lower-dimensional covariates compared to the proposed approach, when certain sparsity structure exists between the treatment assignment and the covariates.  Alternatively, as the proposed reduced covariates are linear combinations of the original covariates, they do not need to be minimal sufficient statistics. Hence, instead of assuming additional sparsity structure in the data, one can also search for finer sufficient statistics in subsequent matching. This can be done in a data-adaptive manner, for example, by using genetic matching based on the proposed reduced covariates.

When the dimension of the covariates is relatively large compared to the sample size, it has been commonly recognized that the conventional SDR methods such as SIR lose effectiveness. In this case, we recommend to use sparse SDR, which constructs the central subspaces on the active set of covariates rather than the original covariates.
For more detail, see \cite{chen2010}.



\bibliographystyle{biom}
\bibliography{myref}

\end{document}